\documentclass[11pt]{article}
\usepackage{fullpage}
\usepackage{amsmath}
\usepackage{amsthm}  
\usepackage{amssymb, algorithm, algpseudocode,eqparbox,thmtools,thm-restate,enumerate,verbatim,bm,array}
\usepackage[usenames,dvipsnames,svgnames,table]{xcolor}
\definecolor{darkgreen}{rgb}{0.0,0,0.9}
\usepackage[pagebackref,letterpaper=true,colorlinks=true,pdfpagemode=none,citecolor=OliveGreen,linkcolor=BrickRed,urlcolor=BrickRed,pdfstartview=FitH]{hyperref}
\usepackage{tikz}

\newcommand{\Input}{\item[{\bf Input:}]}
\newcommand{\Output}{\item[{\bf Output:}]}
\renewcommand{\Return}{\item[{\bf return}]}

\def\equationautorefname~#1\null{%
  equation~(#1)\null
}

\declaretheorem[numberwithin=section]{theorem}
\declaretheorem[sibling=theorem]{lemma}
\declaretheorem[sibling=theorem]{proposition}
\declaretheorem[sibling=theorem]{claim}

\declaretheorem[sibling=theorem]{corollary}

\declaretheorem[sibling=theorem]{fact}
\declaretheorem[sibling=theorem]{definition}
\declaretheorem[sibling=theorem]{example}

\def\setminus{-}

\newenvironment{proofof}[1]{{\medbreak\noindent \em Proof of #1.  }}{\hfill\qed\medbreak}

\def\phiin{\phi_{\text{in}}}
\def\phiout{\phi_{\text{out}}}
\def\bS{{\overline{S}}}

\def\b1{{\bf 1}}
\def\eps{{\epsilon}}
\def\R{\mathbb{R}}
\def\cR{\mathcal{R}}

\def\cL{{\cal L}}
\newcommand{\e}[2]{{w(#1 \to #2)}}

\DeclareMathOperator{\vol}{vol}

\DeclareMathOperator{\poly}{poly}

\DeclareMathOperator{\argmax}{argmax}
\DeclareMathOperator{\argmin}{argmin}

\DeclareMathOperator{\supp}{supp}
\def\inside{\varphi}
\begin{document}

\title{Partitioning into Expanders}
\author{Shayan Oveis Gharan\thanks{Computer Science Division, U.C. Berkeley. This material is  supported by a Stanford graduate fellowship and a Miller fellowship. Email:\protect\url{oveisgharan@berkeley.edu}. }
\and
Luca Trevisan\thanks{Department of Computer Science, Stanford University. This material is based upon  work supported by the National Science Foundation under grants No.  CCF 1017403 and CCF 1216642.
Email:\protect\url{trevisan@stanford.edu}.}
}

\maketitle

\begin{abstract}

Let $G=(V,E)$ be an undirected graph,  $\lambda_k$ be the $k^{th}$ smallest eigenvalue of the normalized laplacian matrix of $G$.
There is a basic fact in algebraic graph theory that $\lambda_k > 0$ if and only if $G$
has at most $k-1$ connected components. We prove a robust version of this fact. If $\lambda_k>0$, then for some $1\leq \ell\leq k-1$, $V$ can be {\em partitioned} into $\ell$ sets $P_1,\ldots,P_\ell$  such that each $P_i$  is a low-conductance set in $G$ and induces a high conductance induced subgraph.
In particular,  $\phi(P_i) \lesssim \ell^3\sqrt{\lambda_\ell}$ and  $\phi(G[P_i]) \gtrsim \lambda_k/k^2$. 

We  make our results algorithmic by designing a simple polynomial time spectral algorithm
to find such partitioning of $G$ with a quadratic loss in the inside conductance of $P_i$'s.
Unlike the recent results on higher order Cheeger's inequality \cite{LOT12,LRTV12},
our algorithmic results do not use higher order eigenfunctions of $G$.
In addition, if there is a sufficiently large gap between $\lambda_{k}$ and $\lambda_{k+1}$, more precisely, if $\lambda_{k+1} \gtrsim \poly(k) \lambda_{k}^{1/4}$ then our algorithm finds a $k$ partitioning of $V$ 
into sets $P_1,\ldots,P_{k}$ such that the induced subgraph $G[P_i]$ has a significantly
larger conductance than the conductance of $P_i$ in $G$. Such a partitioning may represent the best $k$ clustering of $G$.
Our algorithm is a simple local search that only uses the Spectral Partitioning algorithm
as a subroutine. 
We expect to see further applications of this simple algorithm in clustering applications.

Let $\rho(k)=\min_{\text{disjoint } A_1,\ldots, A_k} \max_{1\leq i\leq k} \phi(A_i)$   be the order $k$ conductance constant of $G$, in words, $\rho(k)$ is the smallest value of the maximum conductance of any $k$ disjoint subsets of $V$. 
Our main technical lemma shows that if $(1+\eps)\rho(k)<\rho(k+1)$, then $V$ can be  partitioned into $k$ sets $P_1,\ldots,P_k$ such that
 for each $1\leq i\leq k$, $\phi(G[P_i]) \gtrsim \eps\cdot \rho(k+1)/k$ and $\phi(P_i)\leq k\cdot \rho(k)$.
This significantly improves  a recent result of Tanaka \cite{Tan12} who assumed
an exponential (in $k$) gap between $\rho(k)$ and $\rho(k+1)$. 

\end{abstract}
\thispagestyle{empty}

\newpage
\setcounter{page}{1}

\section{Introduction}
Clustering is one of the fundamental primitives in machine learning and data analysis with a variety of applications 
in information retrieval,
pattern recognition, recommendation systems, etc.
Data clustering may be modeled as a graph partitioning problem,
where one  models each of the data points as a vertex of a graph and
the weight of an edge connecting two vertices represents the similarity of the corresponding data points. We assume the weight is larger if the points are more similar (see e.g. \cite{NJW02}). 
 Let $G=(V,E)$ be an undirected graph with $n:=|V|$ vertices. For all pair of vertices $u,v\in V$ let $w(u,v)\geq 0$ be the weight of the edge between $u$ and $v$ (we let $w(u,v)=0$ if there is no edge between $u$ and $v)$. 

There are several combinatorial measures for the quality of a $k$-way partitioning of a graph
including diameter, $k$-center, $k$-median, conductance, etc.
Kannan, Vempala and Vetta \cite{KVV04} show that several of these measures fail to capture
the natural clustering in simple examples.
Kleinberg \cite{Kle02} show that there is no unified clustering function satisfying three basic properties. 
Kannan et al.~\cite{KVV04} propose conductance as one of the best objective functions for measuring the quality of a cluster.

 For a subset $S\subseteq V$, let the {\em volume} of $S$ be $\vol(S):=\sum_{v\in S} w(v)$,
 where  $w(v) := \sum_{u\in V} w(v, u)$ is the weighted degree of a vertex $v\in V$.
The {\em conductance} of $S$ is defined as
 $$ \phi_G(S):=\frac{
 w(S,\bS)
 }{\vol(S)},$$
 where $\bS=V\setminus S$, and $w(S,\bS)=\sum_{u\in S,v\in\bS} w(u,v)$ is the sum
 of the weight of the edges in the cut $(S,\bS)$. The subscript $G$ in the above definition may be omitted. 
For example, if $\phi(S)=0.1$ it means that $0.9$ fraction 
of the neighbors of a random vertex  of $S$ (chosen proportional to degree) are inside $S$ in expectation. 
The conductance of $G$, $\phi(G)$ is the smallest conductance among all sets that have at most half of the total volume,
$$ \phi(G):=\min_{S: \vol(S)\leq \vol(V)/2} \phi(S).$$
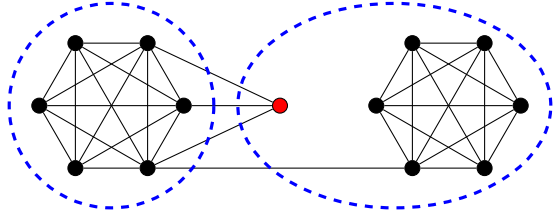
\begin{figure}
\begin{center}
\begin{tikzpicture}[inner sep=2pt,scale=.8,pre/.style={<-,shorten <=2pt,>=stealth,thick}, post/.style={->,shorten >=1pt,>=stealth,thick}]
\tikzstyle{every node} = [draw, circle,fill,black];
\foreach \a/\l in {0/a, 180/b}{
\foreach \b in {0, 60, 120, 180, 240, 300}{
\path (\a:2.8)+(\b:1.2) node [fill=black] (\l_\b){};
}
\foreach \b/\c in {0/60, 0/120, 0/180, 0/240, 0/300, 60/120, 60/180, 60/240, 60/300, 120/180,
120/240, 120/300, 180/240, 180/300, 240/300}{
\path  (\l_\b) edge  (\l_\c);
}
}
\path (0,0) node [fill=red] (o){};
\path (o) edge (b_0) edge (b_300) edge (b_60);
\path (a_240) edge (b_300);
\draw [color=blue,line width=1.2pt,dashed] (-2.8,0) circle (1.7);
\draw [color=blue,line width=1.2pt,dashed] (1.9,0) ellipse (2.6 and 1.7);
\end{tikzpicture}
\caption{In this example although both  sets in the 2-partitioning are of small conductance, 
in a natural clustering the red vertex (middle vertex) will be merged with the left cluster }
\label{fig:phiout}
\end{center}
\end{figure}

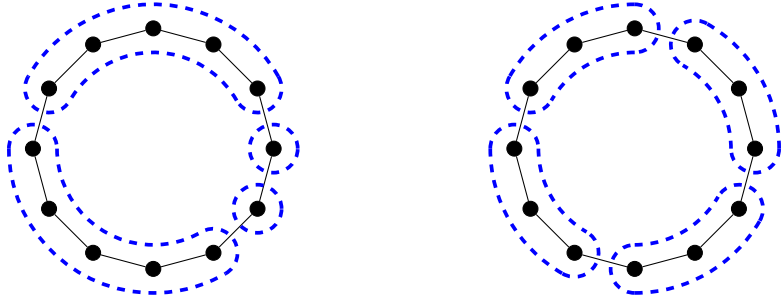
\begin{figure}
\begin{center}
\begin{tikzpicture}[inner sep=2pt,scale=.8,pre/.style={<-,shorten <=2pt,>=stealth,thick}, post/.style={->,shorten >=1pt,>=stealth,thick}]
\tikzstyle{every node} = [draw, circle,fill,black];
\foreach \d/\l in {0/a, 180/b}{
\foreach \a in {0,30, 60, 90, 120, 150, 180, 210, 240, 270, 300, 330}{
\path (\d:4)+(\a:2) node [fill=black] (\l_\a){};
}
\foreach \a/\b in {0/30,30/60, 60/90, 90/120, 120/150, 150/180, 180/210, 210/240, 240/270, 270/300, 300/330, 330/0}{
\path (\l_\a) edge (\l_\b);
}
}
\draw [color=blue,dashed,line width=1.4pt] (b_0) circle (0.4);
\draw [color=blue,dashed,line width=1.4pt] (b_330) circle (0.4);
\foreach \a in {0, 150}{
\begin{scope}[shift={(-4,0)}, rotate=\a]
\draw [color=blue,dashed,line width=1.4pt] (30:2)+(30:0.4) arc (30:150:2.4);
\draw [color=blue,dashed,line width=1.4pt] (30:2)+(30:-0.4) arc (30:150:1.6);
\draw [color=blue,dashed,line width=1.4pt] (30:2)+(30:0.4) arc (30:-150:0.4);
\draw [color=blue,dashed,line width=1.4pt] (150:2)+(150:0.4) arc (150:330:0.4);
\end{scope}
}

\foreach \a in {0, 90, 180, 270}{
\begin{scope}[shift={(4,0)}, rotate=\a]
\draw [color=blue,dashed,line width=1.4pt] (0:2)+(0:0.4) arc (0:60:2.4);
\draw [color=blue,dashed,line width=1.4pt] (0:2)+(0:-0.4) arc (0:60:1.6);
\draw [color=blue,dashed,line width=1.4pt] (0:2)+(0:0.4) arc (0:-180:0.4);
\draw [color=blue,dashed,line width=1.4pt] (60:2)+(60:0.4) arc (60:240:0.4);
\end{scope}
}
\end{tikzpicture}
\caption{Two $4$-partitioning of the cycle graph. In both of the partitionings the number of edges
between the clusters are exactly 4, and the inside conductance of all components is at least 1/2 in both cases. But, the right clustering is a more natural clustering of cycle. }
\label{fig:phiin}
\end{center}
\end{figure}

One approach for constructing a $k$-clustering of $G$ is to find $k$ sets of small conductance.
Shi and Malik \cite{SM00} show that this method provides high quality solutions in image segmentation applications. Recently, Lee et al.~\cite{LOT12} and Louis et al.~\cite{LRTV12}
designed spectral algorithms for finding a $k$-way partitioning where every set
has a small conductance. 
It turns out that in many graphs just the fact that a set $S$ has a small conductance
is not enough to argue that it is a good cluster; this is because although $\phi(S)$ is small, $S$ can be loosely-connected or even disconnected inside (see  \autoref{fig:phiout}). 

Kannan, Vempala and Vetta \cite{KVV04} proposed a bicriteria measure,
where they measure the quality of a $k$-clustering based on the {\em inside} conductance of sets
and the number of edges between the clusters. 
For $P\subseteq V$ let $\phi(G[P])$ be the {\em inside conductance} of $P$, i.e., the conductance
of the induced subgraph of $G$ on the vertices of $P$.
Kannan et al.~\cite{KVV04} suggested that a $k$-partitioning into $P_1,\ldots,P_k$ is good if $\phi(G[P_i])$ is large, and $\sum_{i\neq j} w(P_i,P_j)$ is small.
It turns out that an approximate solution for this objective function can be very different than the ``correct'' $k$-partitioning.  Consider a 4-partitioning of a cycle as we illustrate in \autoref{fig:phiin}.
Although the inside conductance of every set in the left partitioning is within a factor 2 of the right partitioning, the left partitioning does not provide the ``correct'' 4-partitioning of a cycle. 


In this paper we propose a third objective which uses both of the inside/outside conductance of the clusters. 
Roughly speaking, $S\subseteq V$ represents a good cluster
when $\phi(S)$ is small, but $\phi(G[S])$ is large. 
In other words, although $S$ doesn't expand in $G$, the induced subgraph $G[S]$ is an expander. 
\begin{definition}
\label{def:clustering}
We say $k$ disjoint subsets  $A_1,\ldots,A_k$ of $V$ are a $(\phiin,\phiout)$-clustering,
if for all $1\leq i\leq k$, 
$$\phi(G[A_i])\geq \phiin  \ \ \text{ and } \ \ \phi_G(A_i) \leq \phiout.$$
\end{definition}
One of the main contributions of the paper is to study graphs that contain a $k$-partitioning such that $\phiin \gg \phiout$. To the best of our knowledge, the only theoretical result that guarantees a $(\phiin,\phiout)$ partitioning of $G$ is a recent result of Tanaka \cite{Tan12}. 
For any $k\geq 2$, let $\rho(k)$ be the maximum conductance of any $k$ disjoint subsets of $G$,
$$ \rho(k) := \min_{\text{disjoint } A_1,\ldots,A_k} \max_{1\leq i\leq k} \phi(A_i).$$
For example, observe that  $\rho(2) = \phi(G).$
Tanaka \cite{Tan12} proved that if there is a large enough gap between $\rho(k)$
and $\rho(k+1)$ then $G$ has a $k$-partitioning that is a $(\exp(k)\rho(k),\rho(k+1)/\exp(k))$-clustering. 
\begin{theorem}[Tanaka \cite{Tan12}]
\label{thm:tanaka}
If $\rho_G(k+1) > 3^{k+1}\rho_G(k)$ for some $k$, then $G$ has a $k$-partitioning that is a $(\rho(k+1)/3^{k+1},3^k \rho(k))$-clustering.
\end{theorem}
Unfortunately, Tanaka requires a very large gap (exponential in $k$) between $\rho(k)$ and $\rho(k+1)$. Furthermore, the above result is not algorithmic, in the sense that he needs to find the optimum sparsest cut of $G$ or its induced subgraphs to construct the $k$-partitioning.



\subsection{Related Works}
Let $A$ be the adjacency matrix of $G$ and $\cL:=I-D^{-1/2}AD^{-1/2}$ be the normalized laplacian of $G$ with eigenvalues $0=\lambda_1\leq \lambda_2\leq \ldots\lambda_n\leq 2$.
Cheeger's inequality offers the following quantitative connection between $\rho(2)$ and $\lambda_2$:
\begin{theorem}[Cheeger's inequality]
\label{thm:cheegerinequality}
For any graph $G$, 
$$ \frac{\lambda_2}{2} \leq \phi(G)\leq \sqrt{2\lambda_2}.$$
Furthermore, there is a simple near-linear time algorithm (the Spectral Partitioning algorithm) that finds a set $S$ such that
$\vol(S)\leq \vol(V)/2$, and $\phi(S)\leq \sqrt{4\phi(G)}$.
\end{theorem}
The above inequality can be read as follows: a graph $G$ is nearly disconnected if and only if $\lambda_2$ is very close to zero. 
The importance of Cheeger's inequality is that  it does not depend on the size of the graph
$G$, and so it is applicable to massive graphs appearing in practical applications. 

Very recently, Lee et al.~\cite{LOT12} proved 
higher order variants of Cheeger's inequality (see also \cite{LRTV12}). 
In particular, they show that for any graph $G$, $\rho(k)$ very well characterizes $\lambda_k$.
\begin{theorem}[Lee et al.~\cite{LOT12}]
\label{thm:higherordercheeger}
For any graph $G$ and $k\geq 2$,
$$ \lambda_k/2\leq \rho(k) \leq O(k^2)\sqrt{\lambda_k}.$$
\end{theorem}

Meka, Moitra and  Srivastava \cite{MMS13} studied existence of $\Theta(k)$ expander graphs covering most vertices of a graph where the conductance of each expander is a function of $\lambda_k$. 

Kannan, Vempala and Vetta in \cite{KVV04} 
designed an approximation algorithm to find a partitioning of a graph that cuts very few edges
and each set in the partitioning has a large inside conductance. Comparing to \autoref{def:clustering} instead of minimizing $\phi(A_i)$ for each set $A_i$ they minimize 
$\sum_i \phi(A_i)$. 
Very recently, 
Zhu, Lattanzi and Mirrokni \cite{ZLM13} designed a {\em local algorithm} to find a set $S$ such that $\phi(S)$
is small and $\phi(G[S])$ is large assuming that such a set exists. Both of these results do not
argue about the existence of  a partitioning with large inside conductance. Furthermore, unlike Cheeger type inequalities the quality of approximation factor of these
algorithms depends on the size of the input graph (or the size of the cluster $S$).

\subsection{Our Contributions}
\paragraph{Partitioning into Expanders}
There is a basic fact in algebraic graph theory that for any graph $G$ and any $k\geq 2$,
$\lambda_k > 0$ if and only if $G$ has at most $k-1$ connected components.
It is a natural question to ask for a  robust version of this fact. Our main existential  theorem
 provides a robust version of this fact. 

\begin{theorem}
\label{thm:higherorderlambdak}
For any $k\geq 2$ if $\lambda_k>0$, then for some $1\leq \ell \leq k-1$
there is a $\ell$-partitioning of $V$ into sets $P_1,\ldots,P_\ell$ that is a $(\Omega(\rho(k)/k^2), O(\ell\rho(\ell))) = (\Omega(\lambda_k/k^2), O(\ell^3)\sqrt{\lambda_\ell})$ clustering. 
\end{theorem}
The above theorem can be seen as a generalization of \autoref{thm:higherordercheeger}. 

\paragraph{Algorithmic Results}
The above result is not algorithmic   but with some loss in the parameters we can make them algorithmic. 

\begin{theorem}[Algorithmic Theorem]
\label{thm:expanderalgorithm}
There is a simple local search algorithm that for any $k\geq 1$ if $\lambda_k>0$ finds
a $\ell$-partitioning of $V$ into sets $P_1,\ldots,P_\ell$ that is a $(\Omega(\lambda_k^2/k^4), O(k^6\sqrt{\lambda_{k-1}})$
where $1\leq \ell < k$. 
If $G$ is unweighted the algorithm runs in a polynomial time in the size of $G$.
\end{theorem}
\noindent The details of the above algorithm are described in \autoref{alg:kexpander}. 
We remark that the algorithm does not use any SDP or LP relaxation of the problem. It only uses
the Spectral Partitioning algorithm as a subroutine. 
Furthermore, unlike the spectral clustering algorithms studied in \cite{NJW02,LOT12}, our algorithm
does not use multiple eigenfunctions of the normalized laplacian matrix. It rather iteratively refines a 
partitioning of $G$ by adding non-expanding sets that induce an expander.

Suppose that there is a large gap between $\lambda_{k}$ and $\lambda_{k+1}$. Then, the  above
theorem (together with \ref{lem:disjointsets}) implies that there is a $k$ partitioning of $V$ such that inside conductance of each set is significantly larger than its outside conductance in $G$. Furthermore, such a partitioning can be found in polynomial time. This partitioning may represent one of the best $k$-clusterings of the graph $G$.

If
instead of the Spectral Partitioning algorithm we use the $O(\sqrt{\log n})$-approximation
algorithm for $\phi(G)$ developed in \cite{ARV04} the same proof 
implies that $P_1,\ldots,P_\ell$ are a 
$$ \Big( \Omega\Big(\frac{ \lambda_k}{k^2\cdot\sqrt{\log(n)}}\Big), k^3\sqrt{\lambda_{k-1}}\Big)$$
clustering.

To the best of our knowledge, the above theorem provides the first polynomial time algorithm that establishes
a Cheeger-type inequality for the inside/outside conductance of sets in a $k$-way partitioning.

\paragraph{Main Technical Result}
The main technical result of this paper is the following theorem. We show that
even if there is a very small gap between $\rho(k)$ and $\rho(k+1)$ we can  guarantee the existence of a $(\Omega_k(\rho(k+1)),O_k(\rho(k)))$-clustering where we in $\Omega_k(.), O_k(.)$ notations we 
dropped the dependency to $k$.
\begin{theorem}[Existential Theorem]
\label{thm:kexpander}
If $\rho_G(k+1) > (1+\eps) \rho_G(k)$ for some $0<\eps<1$, then 
\begin{enumerate}[i)]
\item There exists $k$ disjoint subsets of $V$ that are a $(\eps\cdot\rho(k+1)/7,\rho(k))$-clustering.
\item There exists a $k$-partitioning of $V$ that is  a $(\eps\cdot\rho(k+1)/(14k),k\rho(k))$-clustering.
\end{enumerate}
\end{theorem}
The importance of the above theorem is that the gap is even independent of $k$ and it can be made arbitrarily close to 0. Compared to \autoref{thm:tanaka},  we require a very small gap between $\rho(k)$ and $\rho(k+1)$ and the quality of our $k$-partitioning has a linear loss in terms of $k$.  
We show tightness of above theorem in \autoref{subsec:tightnessexpanders}.

Using the above theorem it is easy to prove \autoref{thm:higherorderlambdak}.
\begin{proofof}{\autoref{thm:higherorderlambdak}}
Assume $\lambda_k > 0$ for some $k\geq 2$.
By \autoref{thm:higherordercheeger} we can assume $\rho(k) \geq \lambda_k/2>0$.
Since $\rho(1)=0$ we have $(1+1/k)\rho(\ell) < \rho(\ell+1)$ at least for one index $1\leq \ell < k$.
Let  $\ell$ be the largest index such that $(1+1/k) \rho(\ell) <  \rho(\ell+1)$; it follows that
\begin{equation}
\label{eq:diffpartition}
 \rho(k) \leq (1+1/k)^{k-\ell-1}\rho(\ell+1) \leq e\cdot \rho(\ell+1). 
 \end{equation}
Therefore, by part (ii) of \autoref{thm:kexpander} there is a $\ell$-partitioning of $V$ into sets $P_1,\ldots, P_\ell$ such that for all $1\leq i\leq \ell$,
$$ \phi(G[P_i]) \geq \frac{\rho(\ell+1)}{14k\cdot \ell} \geq \frac{\rho(k)}{40k^2} \geq  \frac{\lambda_k}{80k^2}, \text{and} $$
$$ \phi(P_i) \leq \ell\rho(\ell) \leq O(\ell^3)\sqrt{\lambda_\ell}.$$
where we used \eqref{eq:diffpartition} and \autoref{thm:higherordercheeger}.
The following corollary follows.
\end{proofof}

Building on  \autoref{thm:higherordercheeger} we can also prove the existence of a 
good $k$-partitioning of $G$ if there is a large enough gap between $\lambda_k$ and $\lambda_{k+1}$. 
\begin{corollary}
\label{cor:kexpander}
There is a universal constant $c>0$, such that for any graph $G$ if $\lambda_{k+1}\geq c\cdot k^2\sqrt{\lambda_k}$, then there exists a $k$-partitioning
of $G$ that is a $( \Omega(\lambda_{k+1}/k), O(k^3\sqrt{\lambda_k}))$-clustering.
\end{corollary}

\subsection{Tightness of Existential Theorem }
\label{subsec:tightnessexpanders}
In this part we provide several examples showing the tightness of \autoref{thm:kexpander}.
In the first example we show that if there is no gap between $\rho(k)$ and $\rho(k+1)$ then
$G$ cannot be partitioned into expanders. 
\begin{example}
In the first example we construct a graph such that there is no gap between $\rho(k)$ and $\rho(k+1)$ and we show that in any $k$-partitioning there is a set $P$ such that $\phi(G[P]) \ll \rho(k+1)$.
 Suppose $G$ is a star. Then, for any $k\geq 2$, $\rho(k)=1$. 
But, among any $k$ disjoint subsets of $G$ there is a set $P$ with $\phi(G[P])=0$. 
Therefore, for any $k\geq 2$, there is a set $P$ with $\phi(G[P]) \ll \rho(k+1)$.
\end{example}
\noindent In the next example we show that a {\em linear} loss in $k$ is necessary in the quality of our $k$-partitioning in part (ii) of \autoref{thm:kexpander}.
\begin{example}
In this example we construct a graph such that in any $k$-partitioning there is a set $P$
with $\phi(P)\geq \Omega(k\cdot \rho(k))$. Furthermore,
in any $k$ partitioning where  the conductance of every set is $O_k(\rho(k))$,
there is a set $P$ such that $\phi(G[P])\leq O(\rho(k+1)/k)$. 

Let $G$ be a union of $k+1$ cliques $C_0,C_1,\ldots,C_k$ each with $\approx n/(k+1)$ vertices
where $n\gg k$. 
Also, for any $1\leq i\leq k$, include an edge between $C_0$ and $C_i$.
In this graph $\rho(k) = \Theta(k^2/n^2)$ by choosing the $k$ disjoint sets $C_1,\ldots,C_k$.
Furthermore, $\rho(k+1)= \Theta(k\cdot \rho(k))$.

Now consider a $k$ partitioning of $G$. First of all if there is a set $P$ in the partitioning
that contains a proper subset of one the cliques, i.e., $\emptyset \subset (P\cap C_i)\subset C_i$ for some $i$, then $\phi(P) \geq \Omega_k(1/n) = \Omega_k(n\cdot \rho(k))$. Otherwise,
every clique is mapped to one of the sets in the partitioning. Now, let $P$ be the set containing $C_0$ ($P$ may contain at most one other clique). It follows that $\phi(P)=\Omega(k\cdot \rho(k))$.

Now,  suppose we have a partitioning of $G$ into $k$ sets such that the conductance of
each set is $O_k(\rho(k))$. By the arguments in above paragraph none of the sets in the partitioning can have a
proper subset of one cliques. Since we have $k+1$ cliques there is a set $P$ that contains 
exactly two cliques $C_i, C_j$, for $i\neq j$. It follows that $\phi(G[P]) \leq O(\rho(k)/k)$.
\end{example}

\subsection{Notations}
For a function $f:V\to\R$ let 
$$ \cR(f):=\frac{\sum_{(u,v)\in E} w(u,v)\cdot |f(u)-f(v)|^2}{\sum_{v\in V} w(v) f(v)^2}$$
The support of $f$ is the set of vertices with non-zero value in $f$,
$$\supp(f):=\{v\in V: f(v)\neq 0\}. $$
We say two functions $f,g:V\to\R$ are disjointly supported if $\supp(f)\cap \supp(g)=\emptyset$.

For $S\subseteq P\subseteq V$ we use $\phi_{G[P]}(S)$ to denote the conductance of $S$ in  
the induced subgraph $G[P]$. 
For $S,T\subseteq V$ we use 
$$\e{S}{T}:=\sum_{u\in S, v\in T\setminus S} w(u,v).$$
We remark that in the above definition $S$ and $T$ are not necessarily disjoint, so $\e{S}{T}$
is not necessarily the same as $\e{T}{S}$. 

For $S\subseteq B_i\subseteq V$ we define 
$$ \inside(S,B_i):=\frac{\e{S}{B_i}}{\frac{\vol(B_i-S)}{\vol(B_i)}\cdot \e{S}{V\setminus B_i} 
}
$$
Let us motivate the above definition. Suppose $B_i\subseteq V$ such that $\phi_G(B_i)$
is very small but $\phi(G[B_i])$ is very large. Then,  any $S\subseteq B_i$
such that $\vol(S)\leq \vol(B_i)/2$ satisfy the following properties.
\begin{itemize}
\item  Since $\phi_{G[B_i]}(S)$ is large, a large fraction of edges adjacent to vertices of $S$ must
leave this set.
\item Since $\phi_G(B_i)$ is small, a small fraction of edges adjacent to $S$ may leave $B_i$.
\end{itemize}
Putting above properties together we obtain that $\e{S}{B_i} \gtrsim \e{S}{V\setminus B_i}$, thus $\inside(S,B_i)$ is at least a constant. 
As we describe in the next section the converse of this argument is a crucial part of our proof.
In particular, if for any $S\subseteq B_i$, $\inside(S,B_i)$ is large, then $B_i$ has large
inside conductance, and it can be used as the ``backbone'' of our $k$-partitioning.

\subsection{Overview of the Proof}
We prove \autoref{thm:kexpander} in two steps. 
Let $A_1,\ldots,A_k$ be any $k$ disjoint sets such that $\phi(A_i) \leq (1+\eps) \rho(k+1)$.
In the first step we find $B_1,\ldots,B_k$ such that for $1\leq i\leq k$,  $\phi(B_i)\leq \phi(A_i)$ with the crucial property that any subset of $B_i$ has at least
a constant fraction of its outgoing edges inside $B_i$. We then use $B_1,\ldots,B_k$
as the ``backbone'' of our $k$-partitioning. We merge the remaining vertices with $B_1,\ldots,B_k$ to obtain $P_1,\ldots,P_k$ making sure that for each $S\subseteq P_i\setminus B_i$
at least $1/k$ fraction of the outgoing edges of $S$ go to $P_i$ (i.e., $\e{S}{P_i}\geq \e{S}{V}/k$).

We show that if   $2\max_{1\leq i\leq k}\phi(A_i) < \rho(k+1)$ then we can construct $B_1,\ldots,B_k$ such that
every $S\subseteq B_i$ satisfies $\inside(S,B_i)\geq \Omega(1)$ (see \autoref{lem:findTi}). For example, if $\vol(S)\leq \vol(B_i)/2$, we obtain that 
$$ \e{S}{B_i\setminus S} \gtrsim \e{S}{V}. $$
This property  shows that each $B_i$ has an inside conductance of $\Omega(\rho(k+1))$ (see \autoref{lem:Tiexpander}). In addition, it implies that any superset of $B_i$, $P_i\supseteq B_i$, has an inside conductance $\phi(G[P_i]) \gtrsim \alpha\cdot \rho(k+1)$
as long as for any $S\subseteq P_i\setminus B_i$, $\e{S}{B_i} \geq \alpha\cdot \e{S}{V}$ (see \autoref{lem:Piexpander}). 
By latter observation we just need to merge the vertices in $V\setminus B_1\setminus \ldots\setminus B_k$ with $B_1,\ldots,B_k$ and obtain a $k$-partitioning $P_1,\ldots,P_k$
such that for any $S\subseteq P_i\setminus B_i$, $\e{S}{P_i}\geq \e{S}{V}/k$.

\subsection{Background on Higher Order Cheeger's Inequality}
In this short section we use the machinery developed in \cite{LOT12} to show
that for any partitioning of $V$ into $\ell < k$ sets $P_1,\ldots,P_{\ell}$ the minimum inside
conductance of $P_i$'s is $\poly(k) \sqrt{\lambda_k}$.
\begin{theorem} [Lee et al.\cite{LOT12}]
\label{thm:disjointsuppfunctions}
There is a universal constant $c_0>0$ such that for any graph $G=(V,E)$ and $1\leq k\leq n$ there exists $k$ disjointly supported functions
$f_1,\ldots,f_k:V\to\R$ such that for each $1\leq i\leq k$,
$$ \cR(f_i)\leq c_0 k^6\lambda_k.$$
\end{theorem}

\begin{proposition}[Kwok et al.~\cite{KLLOT13}]
\label{prop:lambdakapp}
For any graph $G=(V,E)$, any $k\geq 1$ and any $k$ disjointly supported functions
$f_1,\ldots,f_k:V\to\R$ we have
$$ \lambda_k\leq 2\max_{1\leq i\leq k} \cR(f_i).$$
\end{proposition}

\begin{lemma}
\label{lem:disjointsets}
There is a universal constant $c_0>$ such that for any $k\geq 2$ and any partitioning of $V$ into $\ell$ sets  $P_1,\ldots,P_\ell$ of $V$ where $\ell\leq k-1$,
we have
$$ \min_{1\leq i\leq \ell} \lambda_2(G[P_i]) \leq 2c_0 k^6\lambda_k.$$
where $\lambda_2(G[P_i])$ is the second eigenvalue of the normalized laplacian matrix of the induced graph $G[P_i]$.
\end{lemma}
\begin{proof}
Let $f_1,\ldots,f_k$ be the first $k$ eigenfunctions of $\cL$ corresponding to $\lambda_1,\ldots,\lambda_k$. By definition $\cR(f_i) = \lambda_i$.

By \autoref{thm:disjointsuppfunctions} there are $k$ disjointly supported functions $g_1,\ldots,g_k$ such that
$\cR(g_i)\leq c_0 k^6\lambda_k$. For any $1\leq j\leq \ell$, let $g_{i,j}$ be the restriction of $g_i$ to the induced subgraph
$G[P_j]$. It follows that
\begin{equation}
\label{eq:minrayleigh} \cR(g_i) \geq \frac{\sum_{j=1}^\ell \sum_{(u,v) \in E(P_j)} |g_i(v)-g_i(u)|^2}{\sum_{j=1}^\ell \sum_{v\in P_j} g_i(v)^2} \geq \min_{1\leq j\leq \ell} \frac{\sum_{(u,v)\in E(P_j)} |g_i(u)-g_i(v)|^2}{\sum_{v\in P_j} g_i(v)^2} = \min_{1\leq j\leq \ell} \cR(g_{i,j}).
\end{equation}
For each $1\leq i\leq k$ let $j(i):=\argmin_{1\leq j\leq \ell} \cR(g_{i,j})$.
Since $\ell < k$, by the pigeon hole principle,  there are two indices $1\leq i_1 < i_2 \leq k$
such that $j(i_1)=j(i_2)=j^*$ for some $1\leq j^*\leq \ell$.  
Since $g_1,\ldots,g_k$ are disjointly supported, by \autoref{prop:lambdakapp}
$$ \lambda_2(G[P_{j^*}]) \leq 2 \max\{\cR(g_{i_1,j^*}), \cR(g_{i_2,j^*})\} \leq 2\max\{ \cR(g_{i_1}), \cR(g_{i_2})\} \leq 2c_0 k^6\lambda_k.$$
where the second inequality follows by \eqref{eq:minrayleigh}.
\end{proof}
The above lemma is used in the proof of \autoref{thm:expanderalgorithm}.

\section{Proof of Existential Theorem}
In this section we prove \autoref{thm:kexpander}.
Let $A_1,\ldots,A_k$ are $k$ disjoint sets such that $\phi(A_i)\leq\rho(k)$
for all $1\leq i\leq k$. 
In the first lemma we construct $k$ disjoint sets $B_1,\ldots,B_k$ such that
their conductance in $G$ is only better than $A_1,\ldots,A_k$ with the additional property that
$\inside(S,B_i)\geq \eps/3$ for any $S\subseteq B_i$. 
\begin{lemma}
\label{lem:findTi}
Let $A_1,\ldots,A_k$ be $k$ disjoint sets s.t. $(1+\eps)\phi(A_i) \leq \rho(k+1)$ for $0<\eps<1$. For any $1\leq i\leq k$, there exist a set $B_i \subseteq A_i$ such that the following holds:
\begin{enumerate}
\item $\phi(B_i)\leq \phi(A_i)$.
\item For any $S\subseteq B_i$, 
$ \inside(S,B_i) \geq \eps/3.$ 
\end{enumerate}
\end{lemma}
\begin{proof}
For each  $1\leq i\leq k$ we run  \autoref{alg:findTi} to construct $B_i$ from $A_i$. 
Note that although the algorithm is constructive, it may not run in polynomial time.
The reason is that we don't know any (constant factor approximation) algorithm for
$\min_{S\subseteq B_i} \inside(S,B_i)$. 
\begin{algorithm}
\begin{algorithmic}
\State $B_i=A_i$.
\Loop
\If {$\exists S\subset B_i$ such that  $\inside(S,B_i) \leq \eps/3$},
\State Update $B_i$ to either of $S$ or $B_i\setminus S$ with the smallest conductance in $G$.
\Else
\State {\bf return} $B_i$.
\EndIf
\EndLoop
\end{algorithmic}
\caption{Construction of $B_1,\ldots,B_k$ from $A_1,\ldots,A_k$}
\label{alg:findTi}
\end{algorithm}

First, observe that the algorithm always terminates after at most $|A_i|$ iterations of the loop since  $|B_i|$ decreases in each iteration.
The output of the algorithm always satisfies conclusion 2 of the lemma.  
So, we only need to bound the conductance of the output set. 
We show that  throughout the algorithm we always have 
\begin{equation}
\label{eq:Ticonductance}
\phi(B_i)\leq \phi(A_i).
\end{equation}
In fact, we prove something stronger. That is, the conductance of $B_i$ never increases in the entire run of the algorithm. 
We prove this by induction. At the beginning $B_i=A_i$, so \eqref{eq:Ticonductance} obviously holds. It remains to prove the inductive step.

Let $S\subseteq B_i$ such that $\inside(S,B_i)\leq \eps/3$. 
Among the $k+1$ disjoint sets 
$$\{ A_1,\ldots,A_{i-1}, S, B_i\setminus S, A_{i+1}, A_k\}$$ there is one 
of conductance $\rho_G(k+1)$. So,
$$\max\{\phi(S),\phi(B_i\setminus S)\}\geq \rho_G(k+1) \geq (1+\eps)\phi(A_i) \geq (1+\eps)\phi(B_i). 
$$
The inductive step follows from the following lemma.

\begin{lemma}
\label{lem:minST}
For any set $B_i\subseteq V$ and $S\subset B_i$, if $\inside(S,B_i)\leq \eps/3$ and
\begin{equation}
\label{eq:maxST}
 \max\{\phi(S),\phi(B_i\setminus S)\} \geq (1+\eps)\phi(B_i),
 \end{equation}
then $\min\{\phi(S), \phi(B_i\setminus S)\} \leq \phi(B_i).$
\end{lemma}
\begin{proof}
Let   $T=B_i\setminus S$.  Since $\inside(S,B_i)\leq \eps/3$,
\begin{equation}
\label{eq:implicationpsi}
 \e{S}{T} \leq \frac\eps3\cdot \frac{\vol(T)}{\vol(B_i)}\cdot \e{S}{V\setminus B_i} \leq \frac\eps3\cdot\e{S}{V\setminus B_i}.
 \end{equation}
 We consider two cases depending on whether $\phi(S) \geq (1+\eps)\phi(B_i)$. 
\begin{description}
\item [Case 1: $\phi(S)\geq (1+\eps)\phi(B_i)$.] First, by \eqref{eq:implicationpsi}.
\begin{eqnarray}
\label{eq:eSlower}
(1+\eps)\phi(B_i) \leq \phi(S) = \frac{\e{S}{T} + \e{S}{V-B_i}}{\vol(S)}
\leq \frac{(1+\eps/3) \e{S}{V-B_i}}{\vol(S)}
\end{eqnarray}
Therefore,
\begin{eqnarray*}
\phi(T) &=& \frac{\e{B_i}{V} - \e{S}{V-B_i} + \e{S}{T}}{\vol(T)} \\
&\leq & \frac{\e{B_i}{V} - (1-\eps/3) \e{S}{V-B_i}}{\vol(T)}\\
&\leq& \frac{\phi(B_i)(\vol(B_i) - \vol(S)(1+\eps/2)(1-\eps/3))}{\vol(T)}\\
&\leq & \frac{\phi(B_i)\vol(T) }{\vol(T)} = \phi(B_i).
\end{eqnarray*}
where the first inequality follows by \eqref{eq:implicationpsi} and the second inequality follows by \eqref{eq:eSlower} and that $\eps\leq 1$.

\item [Case 2: $\phi(T)\geq (1+\eps)\phi(B_i)$.] 
First,
\begin{equation}
\label{eq:Tbound}
 (1+\eps)\phi(B_i) \leq \phi(T) = \frac{\e{S}{T} + \e{T}{V-B_i}}{\vol(T)} 
 \end{equation}
Therefore,
\begin{eqnarray*}
\phi(S) &=& \frac{\e{B_i}{V} - \e{T}{V-B_i} + \e{S}{T}}{\vol(S)}\\
&\leq & \frac{\e{B_i}{V} - (1+\eps) \phi(B_i) \vol(T) + 2\e{S}{T}}{\vol(S)}\\
&\leq & \frac{\phi(B_i) (\vol(B_i) - (1+\eps)\vol(T)) + \frac{2\eps}{3} \cdot \vol(T)\cdot\phi(B_i)}{\vol(S)}\\
&\leq & \frac{\phi(B_i)\vol(S)}{\vol(S)}   = \phi(B_i) .
\end{eqnarray*}
where the first inequality follows by \eqref{eq:Tbound}, the second inequality follows by \eqref{eq:implicationpsi} and that $\e{S}{V\setminus B_i}\leq \e{B_i}{V\setminus B_i}$. So we get $\phi(S)\leq \phi(B_i)$. 
\end{description}
This completes the proof of \autoref{lem:minST}.
\end{proof}
\noindent This completes the proof of \autoref{lem:findTi}.
\end{proof}
\noindent Note that sets that we construct in the above lemma do not necessarily define a partitioning of~$G$. 
In the next lemma we show that the sets $B_1,\ldots,B_k$ that are constructed above have large inside conductance.
\begin{lemma}
\label{lem:Tiexpander}
Let  $B_i\subseteq V$,
and  $S\subseteq B_i$
such that $\vol(S)\leq \vol(B_i)/2$. If $\inside(S,B_i),\inside(B_i\setminus S,B_i)\geq \eps/3$ for $\eps\leq 1$, then
$$ \phi_{G[B_i]}(S) \geq \frac{\e{S}{B_i}}{\vol(S)} \geq \frac{\eps}{7}\cdot \max\{\phi(S),\phi(B_i\setminus S)\},$$
\end{lemma}
\begin{proof} 
Let $T=B_i\setminus S$. 
First, we lower bound $\phi_{G[B_i]}(S)$ by $\eps\cdot \phi(S)/7$. Since $\inside(S,B_i)\geq \eps/3$, 
\begin{eqnarray*}
\label{eq:T1subsets}
 \frac{\e{S}{B_i}}{\vol(S)} = \frac{\inside(S,B_i)\cdot \frac{\vol(T)}{\vol(B_i)} \cdot \e{S}{V\setminus B_i}}{\vol(S)}
 \geq  \frac{\eps\cdot \e{S}{V\setminus B_i}}{6\vol(S)}
\end{eqnarray*}
where the first inequality follows by the assumption $\vol(S)\leq \vol(B_i)/2$.
Summing up both sides of the above inequality with $\frac{\eps \e{S}{B_i}}{6\vol(S)}$ and dividing by $1+\eps/6$ we obtain
\begin{eqnarray*}
\frac{\e{S}{B_i}}{\vol(S)} \geq \frac{\eps/6}{(1+\eps/6}\cdot \frac{\cdot \e{S}{V}}{\vol(S)} \geq \frac{\eps\cdot \phi(S)}{7}. 
\end{eqnarray*}
where we used $\eps\leq 1$.
It remains to $\phi_{G[B_i]}(S)$ by $\eps \cdot \phi(B_i\setminus S)/7$. Since $\inside(T,B_i)\geq \eps/3$, 
\begin{eqnarray*}
\frac{\e{S}{B_i}}{\vol(S)} = \frac{\e{T}{B_i}}{\vol(S)} &=& \frac{\inside(T,B_i)\cdot \e{T}{V\setminus B_i}}{\vol(B_i)} \\
&\geq& \frac{\eps}{3}\cdot\frac{\e{T}{V\setminus B_i}}{\vol(B_i)}\\
&\geq& \frac{\eps}{6}\cdot \frac{\e{T}{V\setminus B_i}}{\vol(T)} 
\end{eqnarray*}
where the last inequality follows by the assumption $\vol(S)\leq \vol(B_i)/2$.
Summing up both sides of the above inequality with $\frac{\eps\cdot \e{S}{B_i}}{6\vol(S)}$ we obtain,
\begin{eqnarray*}
(1+\eps/6) \frac{\e{S}{B_i}}{\vol(S)} \geq \frac{\eps}{6} \cdot \frac{\e{T}{V}}{\vol(T)} \geq \frac{\eps\cdot \phi(T)}{6}.
 \end{eqnarray*}
where we used the assumption $\vol(S)\leq \vol(B_i)/2$.
The lemma follows using the fact that $\eps\leq 1$.
\end{proof}
Let $B_1,\ldots,B_k$ be the sets constructed in \autoref{lem:findTi}. Then, for each $B_i$ and $S\subseteq B_i$ since  $\phi(B_j)< \rho(k+1)$ for all $1\leq j\leq k$, 
we get
$$\max(\phi(S),\phi(T))\geq \rho(k+1).$$
Therefore, by the above lemma, for all $1\leq i\leq k$, 
$$\phi(G[B_i]) \geq \eps \cdot \rho(k+1)/7 \text{, and } \phi(B_i)\leq \max_{1\leq i\leq k}\phi(A_i)\leq \rho(k).$$
This completes the proof of part (i) of \autoref{thm:kexpander}.

It remains to prove part (ii). To prove part (ii) we have to turn $B_1,\ldots,B_k$ into a
$k$-partitioning. 
We run the following algorithm to merge the vertices that are not included in $B_1,\ldots,B_k$.
Again, although this algorithm is constructive, it may not run in polynomial time. 
The main difficulty is in finding a set $S\subset P_i\setminus B_i$ such that $\e{S}{P_i}<\e{S}{P_j}$, if such a set exists. 
\begin{algorithm}
\begin{algorithmic}
\State Let $P_i=B_i$ for all $1\leq i\leq k-1$, and $P_k = V\setminus B_1\setminus B_2\setminus \ldots\setminus B_{k-1}$ (note that $B_k\subseteq P_k$).
\While {there is $i\neq j$ and $S\subset P_i \setminus B_i$, such that
$\e{S}{P_i} < \e{S}{P_j}$,}
\State Update $P_i=P_i\setminus S$, and merge $S$ with $\argmax_{P_j} \e{S}{P_j}$.
\EndWhile
\end{algorithmic}
\caption{Construction of $P_1,\ldots,P_k$ based on the $B_1,\ldots,B_k$}
\label{alg:expanderpartition}
\end{algorithm}

First, observe that above algorithm always terminates in a finite number of steps.
This is because in each iteration of the loop the weight of edges between the sets decreases. That is,
$$ \sum_{1\leq i<j\leq k} \e{P_i}{P_j}$$ 
decreases.
The above algorithm has two important properties which are the key ideas of the proof.
\begin{fact}
\label{fact:algprop}
The output of the above algorithm satisfies the following. 
\begin{enumerate}
\item For all $1\leq i\leq k$, $B_i\subseteq P_i$.
\item For any $1\leq i\leq k$, and any $S\subseteq P_i\setminus B_i$, we have
$$ \e{S}{P_i} \geq \e{S}{V}/k.$$
\end{enumerate}
\end{fact}
Next, we use the above properties to show that the resulting sets $P_1,\ldots,P_k$ are non-expanding in $G$
\begin{lemma}
\label{lem:condPi}
Let $B_i\subseteq P_i\subseteq V$ such that 
 $\e{P_i\setminus B_i}{B_i} \geq \e{P_i\setminus B_i}{V}/k$. Then
$$\phi(P_i) \leq k\phi(B_i). $$ 
\end{lemma}
\begin{proof}
 Let $S=P_i\setminus B_i$. 
Therefore,
\begin{eqnarray*} \phi(P_i) = \frac{\e{P_i}{V}}{\vol(P_i)} &\leq& \frac{\e{B_i}{V} + \e{S}{V\setminus P_i} - \e{S}{B_i} }{\vol(B_i)}\\
&\leq & \phi(B_i) + \frac{(k-1) \e{B_i}{S}}{\vol(B_i)} \leq k\phi(B_i).
\end{eqnarray*}
The second inequality uses conclusion 2 of \autoref{fact:algprop}.
\end{proof}

It remains to lower-bound the inside conductance of each  $P_i$. This is proved in the next lemma.
For a $S\subseteq P_i$ we use the following notations in the next lemma (see \autoref{fig:PiTi} for an illustration). 
\begin{equation*}
 \begin{array}{lcl}
 S_B:=B_i \cap S, &  & {\bS}_B := B_i \cap \bS,\\
 S_P:=S \setminus B_i, &  &  {\bS}_P:=\bS\setminus B_i. 
 \end{array}
 \end{equation*}

 \begin{figure}
\begin{center}
\begin{tikzpicture}[inner sep=1.8pt,scale=.8,pre/.style={<-,shorten <=2pt,>=stealth,thick}, 
post/.style={->,shorten >=1pt,>=stealth,thick}]
\draw (0,0) circle (2);
\draw [color=red,line width=1.3pt] (0,2) -- (0,-2) arc (-90:90:2cm);
\draw [color=blue,line width=1.3pt] (2,0) -- (-2,0) arc (180:0:2cm);
\draw  (0.8,0.8) node {$S_B$};
\draw  (-0.8,0.8) node {$\bS_B$};
\draw  (0.8,-0.8) node {$S_P$};
\draw  (-0.8,-0.8) node {$\bS_P$};
\draw [blue] (0,2.5) node {$B_i$};
\draw [red] (2.5,0) node {$S$};
\end{tikzpicture}
\caption{The circle represents $P_i$, the top (blue) semi-circle represents $B_i$ and the right (red) semi-circle represents the set $S$.}
 \label{fig:PiTi}
\end{center}
\end{figure}
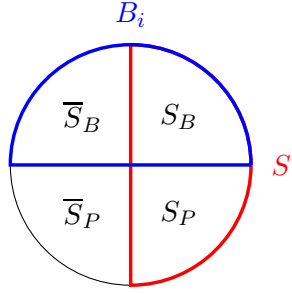

\begin{lemma}
\label{lem:Piexpander}
Let $B_i\subseteq P_i\subseteq V$ and let $S\subseteq P_i$ such that $\vol(S_B)\leq \vol(B_i)/2$.
Let $\rho\leq \phi(S_P)$ and $\rho\leq \max\{\phi(S_B), \phi(\bS_B))\}$ and $0<\eps<1$. If the following conditions hold then
$ \phi(S)\geq \eps\cdot \rho/14k.$
\begin{enumerate}[1)]
\item If $S_P\neq \emptyset$, then $\e{S_P}{P_i} \geq \e{S_P}{V}/k$,
\item If $S_B\neq \emptyset$, then $\inside(S_B,B_i)\geq  \eps/3$ and $\inside(\bS_B,B_i)\geq \eps/3$. 
\end{enumerate}
\end{lemma}
\begin{proof}

We consider 2 cases.
\begin{description} 
\item [Case 1:] {\bf $\vol(S_B)\geq \vol(S_P):$} Because of assumption (2) and $\vol(S_B)\leq \vol(B_i)/2$ we can apply \autoref{lem:Tiexpander}, and we obtain
\begin{eqnarray*}
\phi_{G[P_i]}(S) \geq \frac{\e{S}{P_i}}{\vol(S)}  \geq \frac{\e{S_B}{B_i}}{2\vol(S_B)} \geq \frac{\eps\cdot \max\{\phi(S_B),\phi(\bS_B)\}}{14} \geq \frac{\eps \cdot \rho}{14}.
\end{eqnarray*}

\item [Case 2:] {\bf $\vol(S_P)\geq \vol(S_B):$} 
\begin{eqnarray*}
\phi_{G[P_i]}(S) \geq \frac{\e{S}{P_i}}{\vol(S)} &\geq& \frac{\e{S_P}{P_i\setminus S}+ \e{S_B}{B_i}}{2\vol(S_P)} \\
&\geq & \frac{\e{S_P}{P_i\setminus S}+\eps\cdot \e{S_B}{S_P}/6}{2\vol(S_P)} \\
&\geq & \frac{\eps\cdot \e{S_P}{P_i}}{12\vol(S_P)}\\
&\geq & \frac{\eps\cdot \e{S_P}{ V}}{12k\vol(S_P)}\\
& \geq& \eps\cdot \phi(S_P)/12k \geq \eps\cdot \rho_G(k+1)/12k.
\end{eqnarray*}
where the third inequality follows by the assumption that $\inside(S_B,B_i) \geq \eps/3$ and $\vol(S_B)\leq \vol(B_i)/2$, 
and the fifth inequality follows by assumption (1). 
\qedhere
\end{description}
\end{proof}
Let $B_1,\ldots,B_k$ be the sets constructed in
\autoref{lem:findTi} and $P_1,\ldots,P_k$ the sets constructed in \autoref{alg:expanderpartition},
First, observe that we can let $\rho=\rho(k+1)$. This is because among the $k+1$ disjoint sets
$\{B_1,\ldots,B_{i-1}, S_B, \bS_B,B_{i+1}, B_k\}$ there is a set of conductance $\rho(k+1)$. 
Similarly, among the sets  $\{B_1,B_2,\ldots,B_k,S_P\}$ there is a set of conductance $\rho(k+1)$. Since for all $1\leq i\leq k$,  $\phi(B_i)<\rho(k+1)$, we have $\max\{\phi(S_B),\phi(\bS_B)\} \geq \rho(k+1)$ and
$\phi(P_S)\geq \rho(k+1)$. 
Therefore, by the above lemma,
 \begin{equation*}
\phi(G[P_i]) = \min_{S\subset P_i} \max\{ \phi_{G[P_i]}(S), \phi_{G[P_i]}(P_i\setminus S)\} \geq \eps \cdot \rho(k+1)/14k.
\end{equation*}
This completes the proof of part (ii) of \autoref{thm:kexpander}.

\section{Proof of Algorithmic Theorem}
In this section we prove \autoref{thm:expanderalgorithm}. 
Let 
\begin{equation}
\label{eq:rhosdef}
\rho^* := \min\{\lambda_k/10, 30 c_0 k^5 \sqrt{\lambda_{k-1}}\}.
\end{equation}
where $c_0$ is the constant defined in \autoref{thm:disjointsuppfunctions}.
We use the notation $\phiin := \lambda_k/140k^2$ and $\phiout :=90 c_0 \cdot k^6\sqrt{\lambda_{k-1}}$. 

The idea of the algorithm is simple: we start with one partitioning of $G$, $P_1=B_1=V$.
Each time we try to find a set $S$ of small conductance in one $P_i$. Then, either we can use
$S$  to introduce a new set $B_{\ell+1}$ of small conductance, i.e., $\phi(B_{\ell+1}) \leq 4 \rho^*$, or we can
improve the current $\ell$-partitioning by refining $B_i$ to one of its subsets (similar to \autoref{alg:findTi}) or by moving parts of $P_i$ to the other sets $P_j$ (similar to \autoref{alg:expanderpartition}).


The details of our polynomial time algorithm are described in \autoref{alg:kexpander}. 
Our algorithm  is  a simple local search designed based on \autoref{alg:findTi} and \autoref{alg:expanderpartition}.

\begin{algorithm}[htb]
\begin{algorithmic}[1]
\Input $k>1$ such that $\lambda_k>0$. 
\Output A $(\phiin^2/4,\phiout)$ $\ell$-partitioning of $G$ for some $1\leq \ell<k$.
\State Let  $\ell=1$, $P_1=B_1=V$.
\While {$\exists\ 1\leq i\leq \ell$ such that $\e{P_i\setminus B_i}{B_i} < \e{P_i\setminus B_i}{P_j}$ for $j\neq i$, or Spectral Partitioning finds $S\subseteq P_i$ s.t.
$\phi_{G[P_i]}(S), \phi_{G[P_i]}(P_i\setminus S) < \phiin$ \label{step:partitionwhile}} 
\State Assume (after renaming) $\vol(S\cap B_i) \leq \vol(B_i)/2$.  
\State Let $S_B = S\cap B_i, {\bS}_B = B_i \cap \overline{S}, S_P=(P_i\setminus B_i)\cap S$ and 
${\bS}_P = (P_i\setminus B_i)\cap \bS$ (see \autoref{fig:PiTi}).
\If { $\max\{\phi(S_B), \phi(\bS_B)\} \leq (1+1/k)^{\ell+1} \rho^*$} \label{step:ifmaxSB}
\State Let $B_i=S_B, P_{\ell+1}=B_{\ell+1}=\bS_B$ and $P_i=P_i\setminus \bS_B$. Set $\ell\leftarrow \ell+1$
and {\bf goto} step \ref{step:partitionwhile}. \label{step:addBlB}
\EndIf
\If { $\max\{\inside(S_B, B_i), \inside(\bS_B,B_i)\} \leq 1/3k$, \label{step:ifchangeTi}}
\State Update $B_i$ to either of $S_B$ or $\bS_B$ with the smallest conductance, and {\bf goto}
step \ref{step:partitionwhile}. \label{step:changeTi}
\EndIf
\If {$\phi(S_P) \leq (1+1/k)^{\ell+1} \rho^* $} \label{step:ifmaxSP}
\State Let $P_{\ell+1}=B_{\ell+1}=S_P, P_i=P_i\setminus S_P$. Set $\ell\leftarrow \ell+1$ and {\bf goto} step \ref{step:partitionwhile}. \label{step:addBlP}
\EndIf 
\If { $\e{P_i\setminus B_i}{P_i} < \e{P_i\setminus B_i}{B_j}$ for $j\neq i$, \label{step:ifPimTi} }
\State Update  $P_j=P_j\cup (P_i\setminus B_i)$, and let $P_i=B_i$ and {\bf goto} step \ref{step:partitionwhile}.
\EndIf
\If { $\e{S_P}{P_i} < \e{S_P}{P_j}$ for $j\neq i$, \label{step:ifPSPi}}
\State Update $P_i = P_i-S_P$ and merge $S_P$ with $\argmax_{P_j}\e{S_P}{P_j}$. 
\EndIf
\EndWhile
\Return $P_1,\ldots,P_\ell$.
\end{algorithmic}
\caption{A polynomial time algorithm for partitioning $G$ into $k$ expanders}
\label{alg:kexpander}
\end{algorithm}

Observe that in the entire run of the algorithm $B_1,\ldots,B_\ell$ are always disjoint, $B_i\subseteq P_i$ and $P_1,\ldots,P_\ell$ form an $\ell$-partitioning of $V$.
We prove  \autoref{alg:kexpander} by a sequence of steps.
\begin{claim}
\label{cl:Biupperbound}
Throughout the algorithm we always have
$$ \max_{1\leq i\leq \ell} \phi(B_i) \leq \rho^* (1+1/k)^{\ell}.$$
\end{claim}
\begin{proof} We prove the claim inductively. By definition, at the beginning $\phi(B_1)=0$. 
In each iteration of the algorithm, $B_1,\ldots,B_\ell$ only change in steps \ref{step:addBlB},\ref{step:changeTi} and \ref{step:addBlP}. 
It is straightforward that by executing either of steps \ref{step:addBlB} and \ref{step:addBlP} we  satisfy induction claim, i.e., we obtain $\ell+1$ sets $B_1,\ldots,B_{\ell+1}$ such that  for all $1\leq i\leq \ell+1$,
$$ \phi(B_i)\leq \rho^* (1+1/k)^{\ell+1}.$$

On the other hand, if step \ref{step:changeTi} is executed, then the condition
of \ref{step:ifmaxSB} is not satisfied, i.e., 
$$ \max\{\phi(S_B),\phi(\bS_B)\} > (1+1/k)^{\ell+1}\rho^* \geq (1+1/k) \phi(B_i).$$
where the last inequality follows by the induction hypothesis.
Since $\min\{\inside(S_B,B_i),\inside(\bS_B,B_i)\} \leq 1/3k$ for $\eps=1/k$
by \autoref{lem:minST} we get
$$ \min\{\phi(S_B), \phi(\bS_B)\} \leq \phi(B_i) \leq (1+1/k)^{\ell}\rho^*,$$
which completes the proof.
\end{proof} 
\begin{claim}
\label{cl:ellupperbound}
In the entire run of the algorithm we have $\ell < k$.
\end{claim}
\begin{proof}
The follows from the previous claim. If $\ell=k$, then by previous claim we have
disjoint sets $B_1,\ldots,B_k$ such that
$$ \max_{1\leq i\leq k} \phi(B_i) \leq \rho^*(1+1/k)^k \leq e\cdot \rho^*\leq e\lambda_k/10  < \lambda_k/2.$$
where we used \eqref{eq:rhosdef}. But, the above inequality implies $\rho(k) < \lambda_k/2$ which
contradicts \autoref{thm:higherordercheeger}.
\end{proof}
\begin{claim}
\label{cl:algoutput}
If the algorithm terminates, then it returns a $\ell$-partitioning of $V$ that is a 
$(\phiin^2/4, \phiout)$-clustering.
\end{claim}
\begin{proof}
Suppose the algorithm terminates with sets  $B_1,\ldots,B_ \ell$ and $P_1,\ldots,P_\ell$.
Since by the loop condition, for each $1\leq i\leq \ell$, 
$$\e{P_i\setminus B_i}{B_i} \geq \e{P_i\setminus B_i}{V}/\ell,$$ 
by \autoref{lem:condPi},  
$$\phi(P_i)\leq  \ell \phi(B_i) \leq \ell \cdot e\cdot \rho^* \leq 90 c_0 \cdot  k^6 \sqrt{\lambda_{k-1}}.$$
where the second inequality follows by \autoref{cl:Biupperbound}, and the last inequality follows by \autoref{cl:ellupperbound} and \eqref{eq:rhosdef}.

On the other hand, by the condition of the loop and the performance of Spectral Partitioning algorithm as described in \autoref{thm:cheegerinequality}, for each
$1\leq i\leq k$,
$$ \phi(G[P_i]) \geq \phiin^2/4 = \Omega( \lambda_k^2/k^4). $$
\end{proof}

It remains to show that the algorithm indeed terminates. First,  we show that in each iteration of the loop at least one of the conditions
are satisfied.
\begin{claim}
In each iteration of the loop at least one of the conditions hold.
\end{claim}
\begin{proof}
We use \autoref{lem:Piexpander} to show that if none of the conditions in the loop are satisfied then $\phi(S) \geq \phiin$ which is a contradiction. So, for the sake of contradiction assume in an iteration of the loop none of the conditions hold.


First, since conditions of \ref{step:ifchangeTi} and \ref{step:ifPSPi} do not hold, for $\eps=1/k$ assumptions (1) and (2) of \autoref{lem:Piexpander} are satisfied.
Furthermore, since condition of steps \ref{step:ifmaxSB} and \ref{step:ifmaxSP} do not hold
\begin{align*} 
\max\{\phi(S_B,\bS_B)\} &= \max\{\phi(B_1),\ldots,\phi(B_{i-1}),\phi(S_B),\phi(\bS_B),\phi(B_{i+1},\ldots,\phi(B_\ell)\} \geq \max\{\rho^*,\rho(\ell+1)\}. \\
\phi(S_P) &= \max\{\phi(B_1),\ldots,,\ldots,\phi(B_\ell), \phi(S_P)\} \geq \max\{\rho^*,\rho(\ell+1)\}. 
\end{align*}
where we used \autoref{cl:Biupperbound}.
So, for $\rho =\rho^*$  and $\eps=1/k$ by \autoref{lem:Piexpander} we get
\begin{equation}
\label{eq:rholowerbound}
 \phi(S) \geq \frac{\eps\cdot \rho}{14 k} = \frac{\max\{\rho^*,\rho(\ell+1)\}}{14k^2}. 
 \end{equation} 
 Now, if $\ell=k-1$, then by \autoref{thm:higherordercheeger}
 we get
 $$ \phi(S) \geq \frac{\rho(k)}{14k^2} \geq \frac{\lambda_k}{28k^2} \geq \phiin,$$
 which is a contradiction and we are done.
 Otherwise,  we must have $\ell<k-1$. Then by \autoref{lem:disjointsets}, 
\begin{equation}
\label{eq:phisupperbound}
 \phi(S) \leq \min_{1\leq i\leq \ell} \sqrt{2\lambda_2(G[P_i])} \leq \sqrt{4c_0 k^6\lambda_{k-1}}, 
 \end{equation}
where the first inequality follows by the Cheeger's inequality (\autoref{thm:cheegerinequality}),
Putting \eqref{eq:rholowerbound} and \eqref{eq:phisupperbound} together we have
$$ \rho^* \leq 14k^2\sqrt{4c_0 k^6 \lambda_{k-1}}. $$
But, by definition of $\rho^*$ in equation \eqref{eq:rhosdef}), we must have $\rho^*=\lambda_k/10$.
Therefore, by \eqref{eq:rholowerbound}, 
$$ \phi(S) \geq \frac{\lambda_k}{140k^2} =\phiin,$$
which is a contradiction, and we are done.
\end{proof}

It remains to show that the algorithm actually terminates and if $G$ is unweighted it terminates in polynomial time.
\begin{claim}
For any graph $G$ the algorithm terminates in finite number of iterations of the loop.
Furthermore, if $G$ is unweighted, the algorithm terminates after at most $O(kn\cdot |E|)$ iterations of the loop.
\end{claim}
\begin{proof}
In each iteration  of the loop at least one of  conditions in lines \ref{step:ifmaxSB},\ref{step:ifchangeTi},\ref{step:ifmaxSP},\ref{step:ifPimTi} and \ref{step:ifPSPi} are satisfied. 
By \autoref{cl:ellupperbound}, Lines \ref{step:ifmaxSB} and \ref{step:ifmaxSP} can be satisfied at most $k-1$ times. Line \ref{step:ifchangeTi} can be satisfied at most $n$ times (this is because each time the size of  one $B_i$ decreases by at least one vertex). Furthermore, for a fixed $B_1,\ldots,B_k$, \ref{step:ifPimTi},\ref{step:ifPSPi} may hold
only finite number of iterations, because each time the total weight of the edges between $P_1,\ldots,P_k$ decreases. 
In particular, if $G$ is unweighted, the latter can happen at most $O(|E|)$ times. So, for undirected graphs  the algorithm terminates after at most $O(kn\cdot |E|)$ iterations of the loop.
\end{proof}

\noindent This completes the proof of \autoref{thm:expanderalgorithm}.

\section{Concluding Remarks}
We propose a new model for measuring the quality of $k$-partitionings of graphs which involves
both  the inside and the outside conductance of the sets in the partitioning.
We believe that this is often an accurate model of the quality of solutions in practical applications.
Furthermore, the simple local search \autoref{alg:kexpander}
can be used as a pruning step at the end of any graph clustering algorithm.

From a theoretical point of view, there has been a long line of works on the sparsest cut problem and partitioning of a graph into sets of small outside conductance \cite{LR99,LLR95,ARV04,ALN05} but none of these
works study the inside conductance of the sets in the partitioning. We think it is a fascinating open problem to study efficient algorithms based on linear programming or semidefinite programming relaxations that provide a bicriteria approximation to the $(\phiin,\phiout)$-clustering problem.

Several of our results can be generalized or improved. In \autoref{thm:kexpander} we significantly improve \autoref{thm:tanaka} of Tanaka \cite{Tan12} and we show that
even if there is a small gap between $\rho(k)$ and $\rho(k+1)$, for some $k\geq 1$, then the graph admits a $k$-partitioning that is a $(\poly(k)\rho(k+1),\poly(k)\rho(k))$-clustering. Unfortunately, to carry-out this result to the domain of eigenvalues we need to look for a significantly larger gap between $\lambda_k,\lambda_{k+1}$ (see \autoref{cor:kexpander}). It remains an open problem if  such a partitionings of $G$ exists under only a constant gap between $\lambda_k,\lambda_{k+1}$.

\subsection*{Acknowledgements}
We would like to thank anonymous reviewers for helpful comments.
Also, we would like to thank Pavel Kolev for a careful reading of the paper and exclusive comments.
\bibliographystyle{alpha}
\bibliography{references}

\end{document}